\documentclass[smallextended,final]{svjour3}
\usepackage{mathptmx}

\usepackage[utf8]{inputenc}
\usepackage{bbm}
\usepackage{graphicx}
\usepackage{natbib}
\usepackage{url}

\usepackage[english]{babel}
\usepackage{boxedminipage}
\usepackage{stmaryrd}

\usepackage{amsmath}
\usepackage{amsfonts}
\usepackage{amssymb}

\newcommand{\SeenInts}{\mathcal{S}}
\newcommand{\Lattice}[1]{\mathcal{S}_{#1}}
\newcommand{\Compatible}[1]{\mathcal{C} \hspace*{-3pt} \left( #1 \right) }
\newcommand{\floor}[1]{\lfloor #1 \rfloor}
\newcommand{\interval}[1]{\llbracket #1 \rrbracket}
\newcommand{\maxS}{\hat{x}}
\newcommand{\maxSi}{\hat{n}}

\newcommand{\bbN}{\ensuremath{\mathbb{N} }}

\newcommand{\bbP}{\ensuremath{\mathbb{P} }}
\newcommand{\bbE}{\ensuremath{\mathbb{E} }}

\newcommand{\card}{\mathrm{card}}

\begin{document}

\title{Semiparametric Estimation of a Noise Model with Quantization Errors}
\author{Li-Thiao-Té Sébastien}
\institute{CMLA ENS Cachan CNRS UniverSud,  61 Avenue President Wilson, F-94230 Cachan \email{lithiaote@cmla.ens-cachan.fr}}

\maketitle

\begin{abstract}
The detectors in mass spectrometers are precise enough to count ion events, but in practice large quantization errors affect the observations. To study the statistics of low intensity chemical noise, we model the detector signal as $X = \floor{\tau \, N}$ and estimate both $\tau$ and $N$ in a semi-parametric approach where the integer valued random variable $N$ represents the number of ions and $\tau$ represents the gain parameter of the detector. When $\tau \leq 1$, we explain why the gain parameter cannot be recovered without a priori information on $N$. When $\tau > 1$ however, $N$ can be deduced from $X$ and a sufficiently precise estimate of $\tau$. To perform parametric estimation of $\tau$, we first study simple estimators which provide useful upper bounds. We then introduce the concept of \emph{compatible lattices} and we derive an optimal estimator that is independent of the law of $N$. 
\end{abstract}


%%%%%%%%%%%%%%%%%%%%%%%%%%%%%%%%%%%%%%%%%%%%%%%%%%%%%%%%%%%%%%%%%%%%%%%%%%%%%%%%%%%%%%%%%%%%%%%%%%%%
\section{Introduction}

\subsection{Ion Detectors}

\paragraph{}
Mass spectrometers are instruments that ionize the compounds of a sample, separate the ions, then quantify the ions at each mass to charge ratio. The resulting signal is a histogram that represents ion intensity as a function of the mass to charge ratio of the ions. With sufficient precision in the separation and the mass to charge measurement, the components of the sample can be identified and quantified. Mass spectrometers are widely used for analysing very diverse mixtures, e.g.\ detecting explosives for airport security or analysing oil products. See \cite{aebersold2003msb, lane2005msb, siuzdak1996msb} for an introduction to mass spectrometry in the life sciences.

\paragraph{}
We consider detectors similar to microchannel plate detectors that are used in most mass spectrometers \cite{wiza1979mpd}. When an ion hits the detector plate, it produces an analog signal that is amplified, quantized, then reported to the computer. The level of quantization is quite high as there may be only $2^{11} = 2048$ levels\footnote{Single precision floating point numbers have $2^{24} \sim 16.10^6$ levels of precision plus sign and exponent.} in some instruments, and small signals as well as chemical noise are strongly affected by quantization effects.

% Nuclear Instruments and Methods, Vol. 162, 1979, pages 587 to 601
% MICROCHANNEL PLATE DETECTORS
% JOSEPH LADISLAS WIZA

% Agilent vend des convertisseurs high speed avec seulement 12bits 
% http://www.home.agilent.com/agilent/product.jspx?nid=-35502.0.00&cc=US&lc=eng

\paragraph{}
Specific difficulties have appeared with high-throughput analyses of biological material. In particular, biological samples may contain trace amounts of molecules of interest. These are difficult to distinguish from chemical noise which produces patterns similar to real signals \cite{andreev2003uda,krutchinsky2002ncn,windig1996nab}. In 2004, \cite{anderle2004qrd} suggested Poisson-like behaviour for the ion intensity based on a linear relationship between the mean and variance of the noise. This linear relationship suggests that the amplification factor of the detector may be unaccounted for in the data set. 

% Statistical Characterization of Chemical Noise in MALDI TOF MS by Wavelet Analysis of Multiple Noise Realizations
% Hyunjin Shin, M.S.,1 Mehul P. Sampat, M.S.,2 Sheldon F. Bish,3 John M. Koomen, Ph.D.,4 and Mia K. Markey, Ph.D.2

% On the nature of the chemical noise in MALDI mass spectra
% Authors: Krutchinsky A.N.; Chait B.T.1

\paragraph{}
To study chemical noise in the experimental data, we interpret the amplification factor as an overdispersion parameter in a semi-parametric approach. 
To study chemical noise in the experimental data, we estimate the amplification factor and an unknown distribution for the chemical noise in a semi-parametric approach. Let $N$ denote the number of chemical noise ions that reach the detector, we consider the following observation model: 
\[ X = \floor{ \tau N + \varepsilon} \]
where the noisy signal $\tau N + \varepsilon$ is truncated before observation. $\tau$ represents the amplification factor of the detector and $\varepsilon$ represents electronic noise. In this paper, we make the assumption that $\varepsilon = 0$ or equivalently that there are only quantization errors in the measurements. The observation model is associated with the statistical structure $(\bbN, \mathfrak B, \bbP_{\theta})$, with $\theta = (\tau,N)$ where $\tau$ is a positive real number and $N$ is a probability distribution on $\bbN$. 

\paragraph{}
We believe a priori that $N$ is Poisson distributed as it models rare events (ion counts). Consequently, $\tau$ can be interpreted as an overdispersion parameter affecting Poisson distributed observations. This has been tackled in the framework of double exponential families, as presented by Efron in \cite{efron1986def}. In \cite{antoniadis1999esp}, Antoniadis et al use double exponential families in a regression model to analyse diffraction spectra. This corresponds to estimating the regression function $\mu_i$ in the model $X'_i = \tau N_i$ where $N_i$ is Poisson distributed with varying parameter $\mu_i$. However, this framework does not explicitly take into account quantization errors and thus provides poor parameter estimates as we will show in Section \ref{doublepoisson}. Moreover, we wish to confirm the Poisson hypothesis using non parametric estimation.

\paragraph{}
Our approach is to first estimate $\tau$ given a set of observations of $X$, then deduce the distribution of $N$ from the estimate. We show that the estimate is precise enough to allow complete disambiguation of the observations.

% 												% Structure paramétrique

\subsection{Estimation of the Ion Statistics}

\paragraph{}
With negligible quantization error, the observation model becomes $X = \tau N$. Estimation of $\tau$ is trivial; all that is required is to observe the event $\{N = 1\}$ i.e.\ $\{X = \tau \times 1\}$, or the two events $\{x_1 = \tau i\}$ and $\{x_2 = \tau (i+1)\}$ and compute the difference $x_2 - x_1$. To recover $N$, it then suffices to consider $X / \tau$. The quantization error may be neglected when $\tau \gg 1$ in the observation model 
$X = \floor{\tau \, N}$ and the previous estimates provide $\tau$ with a precision on the order of the quantization error.

\paragraph{}
In the general case, we can recover the samples of $N$ from the samples of $X$ when the mapping  $x \mapsto \floor{ \tau \, x }$ is injective. The inverse mapping is $ y \mapsto \lceil y / \tau \rceil$. We call this situation the \emph{distinguishible} case. It occurs if and only if $\tau \geq 1$ (see proof in the Appendix, Prop \ref{injective}). In this situation, the semi-parametric approach can be separated into parametric estimation of the gain parameter $\tau$ then non parametric estimation of the distribution of $N$ from iid samples.

\begin{example}
\label{example2}  \mbox{} \vspace{-8pt} 
\begin{verbatim}
> data = floor(1.32 * n)  
% Distinguishible case
> n
 [1]  1  2  3  4  5  6  7  8  9 10
> data
 [1]  1  2  3  5  6  7  9 10 11 13
\end{verbatim}
\end{example}

\paragraph{}
When $\tau$ is smaller than $1$, the truncation error merges adjacent values of $N$. In the following example, the events $\{N = 3\}$ and $\{N = 4\}$ cannot be distinguished in the data set. This is because the corresponding observation is $\{X = 2\}$ in both cases.
\begin{example}
\label{example1}  \mbox{} \vspace{-8pt} 
\begin{verbatim}
> data = floor(0.68 * n) 
% Non distinguishible case
> n
 [1]  1  2  3  4  5  6  7  8  9 10
> data
 [1]  0  1  2  2  3  4  4  5  6  6
\end{verbatim}
\end{example}

% \begin{defi}[Distinguishibility]
% The model is \emph{distinguishible} if the mapping $x \longrightarrow \lfloor \tau \, x \rfloor$ is injective. The model is \emph{non distinguishible} otherwise.
% \end{defi}
% 
% \begin{proposition}[Distinguishibility in the model]
% The model $X = \lfloor \tau \, N \rfloor$ is distinguishible if and only if $\tau \geq 1$.
% \end{proposition}
% \begin{proof}
% If $\tau = 1$, the mapping is the identity function. Suppose $\tau > 1$ and let $n_1$ and $n_2$ denote two positive integers such that $n_1 < n_2$. Then $\tau n_2 - \tau n_1 > \tau > 1$ and $\floor{\tau n_2} > \floor{\tau n_1}$. When $\tau < 1$, the mapping is not injective because $\floor{\tau \times 1} = 0$. 
% \end{proof}
% \textbf{Remark}
% The pf suggests that removing 0 from the data set may help obtain the distinguishibility. That is not the case because when $\tau < 1$, there exist infinitely many positive integers $n$ such that $ \tau n - \floor{\tau n} < 1-\tau $; this is because either $\tau$ is a rational number $p / q$ and any number in $q \bbN $ is suitable, or $\tau$ is irrational and the sequence $\tau n$ is uniformly distributed modulo 1. For such $n$, $\tau (n+1) - \floor{\tau n} < 1$ and therefore $\floor{\tau (n+1)} = \floor{\tau n}$.

\paragraph{}
In the distinguishible case, it is natural to sort and index the observed values in order to determine the mapping $x \mapsto \lfloor \tau \, x \rfloor$. This is not sufficient in practice because of missing values or outliers which can modify the indexes.

\subsection{Observation Set}

\paragraph{}
The gain parameter and the law of $N$ have separate effects on $X$. In the distinguishible case, the distribution function of $X$ is a transformation of the distribution function of $N$ by the mapping $x \mapsto \lfloor \tau \, x \rfloor$. The gain parameter and the truncation error only distort the position of each peak, whereas the relative frequencies are unchanged. Consequently, the support $\SeenInts$ of the empirical distribution is sufficient information for estimating $\tau$ whereas the empirical frequencies are sufficient information for the distribution of $N$.

\paragraph{}
In the non distinguishible case, the set of observed integers is always $\bbN$ for large samples (see Section \ref{compatible_values}). As a consequence, $\tau$ cannot be estimated based on that set alone. A semi-parametric approach is not feasible either. For instance, we cannot estimate the mean $\bbE[N]$ but only $\bbE[X] = \tau \bbE[N]$. To separate $\tau$ and $N$, we have to provide prior assumptions on the distribution of $N$ like a Poisson parametric family.

\paragraph{}
In the following, we study properties of the set $\SeenInts$ of observed integers. This set can be constructed from the dataset in $\mathcal{O}(n \log(n))$  time using sorting for example. The algorithmic complexity of the following algorithms is governed by the size of $\SeenInts$, and in particular, the maximum integer in $\SeenInts$. 

\paragraph{}
We focus on the distinguishible case, and perform parametric estimation of the gain parameter from a random set of integers. As the support of the empirical distribution is a sufficient statistic for $\tau$, we use the statistical structure $\left( \Omega = 2^\bbN, \mathfrak T, \bbP_\tau, \tau \in \left] 1,+\infty \right[ \, \right)$ where $\Omega$ is the power set of $\bbN$ and $\mathfrak T$ is the exhaustive $\sigma$-algebra on $\Omega$. \\
$\left( \bbP_\tau, \tau \in \left] 1,+\infty \right[ \, \right)$ is a parametric family of distributions on $\Omega$ that is implicitly generated in the following way.
For a fixed integer $n$ and fixed but unknown integer-valued random variable $N$, $\bbP_\tau$ is the distribution of the random variable $\SeenInts$ which is the set of observed integers in an independent identically distributed sample $(X_1, \ldots, X_n)$ of $X = \floor{\tau N}$.

% \textbf{Remark}
% Note that independence of the observations is not a required hypothesis. The procedures developped here in the context of repeated experiments (ion counts) can be used when $N$ is an integer-valued random process with unknown correlation structure.

\subsection{Organization of the paper}

% The distinguishible and non distinguishible scenarios correspond to two very different instances of the model. As will be made clearer in Section \ref{compatible_values}, estimation of $\tau$ is impossible in the non distinguishible case without knowledge of the distribution of $N$. In contrast, an estimation procedure with very little prerequisite on $N$ is presented in this paper in the distinguishible case, and we focus on that case henceforth.

\paragraph{}
To estimate $\tau$ in the distinguishible case, we first provide simple estimators for $\tau$ in Section \ref{bornes_obs}. These are later used as a starting point for improved estimators and to restrict the search space for $\tau$.

\paragraph{}
In Section \ref{compatible_values}, we define the notion of \emph{compatible values} and provide a few properties of the set of compatible values. In particular, the true parameter $\tau$ is a compatible value and is close to the highest compatible value. This leads to an optimal estimator that is described in \ref{implementation}. 

\paragraph{}
We show the results of some simulations in section \ref{simulations} and compare with the Maximum Likelihood Estimator obtained from the Double Poisson Family, an estimator based on linear regression and another one based on Fourier transform.

% \subsection{Preliminary remarks}
% 
% \textbf{Support of $N$}
% Let $N$ be a random variable such that $\bbP[N=i]=0$ whenever $i$ is odd. Then the model $(\tau,N)$ cannot be distinguished from $(2\tau,N/2)$ which would be preferred in practice (parsimony principle). In general, the reasonable estimate of the period is $pgcd(N)\tau$. However, Sadler05 prove that the probability of $pgcd(N)>1$ is small.
% 
% Results can be extended to the case $N \in \bbZ$. However, it appears more useful to consider the case of an integer-valued random process $N$ and the following observation model : 
% \[X = \lfloor \tau N + \Phi \rfloor\]
% In this model, a periodic signal with missing observations is quantized in time. This situation may appear when studying spatially periodic data with pixel-based detectors.
% 
% Results can be extended to the case of rounding instead of truncation.

%%%%%%%%%%%%%%%%%%%%%%%%%%%%%%%%%%%%%%%%%%%%%%%%%%%%%%%%%%%%%%%%%%%%%%%%%%%%%%%%%%%%%%%%%%%%%%%%%%%%
\section{Estimators and Upper Bounds for $\tau$}
\label{bornes_obs}

% \subsection{Local Bounds}

% \paragraph{}
% In this section, we study the random set $\SeenInts$ and derive estimators for $\tau$ by counting the number of elements in $\SeenInts$· We propose three estimates that turn out to be also upper bounds for $\tau$.
% 

% \emph{Any} two elements in the data set can be used, although close values in the data set provide a tighter bound. This should be only used for visual inspection of the data because the following results provide tighter bounds with the same linear computational load.
% 
% \textbf{Remark}
% This bound is useful because as soon as two consecutive integers are observed, $\tau$ is less than 2. Conversely, if $\tau$ is greater than 2, there cannot be consecutive integers. Proposition \ref{any_two_obs} can thus be used to distinguish quickly whether we are close to the distinguishible case or not. 
% In the rest of the section, we focus on the case $\tau < 2$ where the set $\SeenInts$ contains many integers. When $\tau \geq 2$, the observed values are separated

% When $\tau < 2$, the set of observed integers is ``dense'' in the sense that there may be several consecutive integers in $\SeenInts$. In that case the previous bound can be improved as follows.
% 
\paragraph{}
The results in this section are based on the following idea. Two points in $\SeenInts$ are separated by at least $\floor{\tau}$. Consequently, when $\tau$ is large, then $\SeenInts$ is a sparse set, whereas $\SeenInts$ is dense when $\tau$ is near 1. For instance, there are consecutive points in $\SeenInts$ if and only if $\tau \leq 2$ (see Proposition \ref{any_two_obs} in the Appendix).

\paragraph{}
A better estimate can be obtained by combining more than 2 consecutive points. Let $\interval{x,y}$ denote the set of integers between $x$ and $y$. If $\interval{x,y}$ is a subset of $\SeenInts$, then $\displaystyle \tau < 1 + \frac{1}{y-x}$.
Consequently, $\tau$ can be estimated by $\displaystyle 1+\frac{1}{y-x}$ with a precision on the order of the inverse of the length of the interval $\displaystyle \frac{1}{y-x}$. However, this estimator is strongly affected by missing values in $\SeenInts$.

\paragraph{}
Instead of considering all the segments in $\SeenInts$, we propose to use the overall density of the set, which is easier to compute algorithmically.
Let $\maxS = \lfloor \tau \; \maxSi \rfloor$ denote the largest integer in $\SeenInts$. Then $\displaystyle \tau < \frac{\maxS + 1}{\maxSi}$. When $\maxSi$ is unknown (because of potential missing values), let $n$ denote the number of non zero observed integers i.e.\ the number of elements in $\SeenInts$. Then 
\[
\label{upperBound}
 \tau < \frac{\maxS + 1}{\maxSi} \leq \frac{\maxS + 1}{n}.
\]

\paragraph{}
Consequently, $\displaystyle \frac{\maxS + 1}{n}$ is an estimate of $\tau$ with precision on the order of $1/n$ (without missing values). As it uses the whole data, it is usually more precise than the previous bound. We will use this in the rest of the paper to restrict the search space for $\tau$.

% \subsection{Example}
\paragraph{}
Let us compare the previous bounds on an example. Suppose that $\tau = 1.32$ and $\SeenInts = \{1,2,3,5,6,7,9,10,11,13 \}$. \\
As there are consecutive integers in $\SeenInts$ we obtain $\tau < 2$ . \\
Using the interval $\interval{5,7}$, we obtain $\tau < 1+1/2$. \\
Using the interval $\interval{9,11}$, we obtain $\tau < 1+1/2$ as well. \\
The density upper bound is $\tau < 14 / 10$.

\paragraph{Remark}
We only provided upper bounds in this section because lower bounds can only be deduced from the integers that cannot be generated in the model. These are difficult to distinguish from missing values, which are integers that can be generated in the model, but do not appear in the set $\SeenInts$ of observed integers. 

%%%%%%%%%%%%%%%%%%%%%%%%%%%%%%%%%%%%%%%%%%%%%%%%%%%%%%%%%%%%%%%%%%%%%%%%%%%%%%%%%%%%%%%%%%%%%%%%%%%%
\section{Compatible Values}
\label{compatible_values}

\paragraph{}
The upper bounds that we proposed in the previous section are easy to compute but rather poor because they only take into account the proportion of observed integers. In this section, we describe an algorithm with higher computational load but which can leverage the information in the location of each observed integer in the data set.

\subsection{Lattices of Integers}

% \begin{defi}[Empirical Lattice]
% Let $X_n$ denote a sample of size $n$ of $X = \lfloor \tau \, N \rfloor$. We call the \emph{empirical lattice} generated by $X_n$ the set of all the integers observed in the dataset and denote that set by $\mathcal{G}_e$.
% \end{defi}

% \begin{definition}[Lattice associated to a positive real number]
% The lattice associated to $t \in \bbR^+_*$ is the infinite set of integers $\Lattice{t} = \{ x = \lfloor t \, k \rfloor, k \in \bbN \}$. The set $\SeenInts$ of observed values is also called the empirical lattice.
% \end{definition}
% 
% \begin{definition}[Compatibility]
% Let $A$ and $B$ denote two (arbitrary) sets of integers. We say that $A$ is compatible with $B$ when $B \subset A$. In particular, we will say that a real number $t$ is compatible with the data (resp. $t_1$ is compatible with $t_2$) when the empirical lattice is included in the lattice $\Lattice{t}$ (resp. when $\Lattice{t_2} \subset \Lattice{t_1}$).
% \end{definition}
% 
% \paragraph{}
% Note that the compatibility relation is intrinsically unidirectional, in particular because the empirical lattice is finite, will miss integers and may contain outliers.

\paragraph{}
In the observation model $X = \floor{\tau N}$ where $N$ is integer valued, only specific integers can be generated. Given a strictly positive real number $t$, let us define the set of possible values for $x$ as the \emph{lattice associated to $t$}, i.e.\ the infinite set of integers $\Lattice{t} = \{ x = \lfloor t \, k \rfloor, k \in \bbN \}$. The set of observed integers $\SeenInts$ is also called the \emph{empirical lattice}.

\paragraph{}
With infinitely many observations, the parameter $\tau$ is completely characterized by the empirical lattice as the following proposition shows. This justifies that $\SeenInts$ is sufficient information for estimating $\tau$.

\begin{proposition}[Equivalence between lattices and numbers]
\label{equiv_grid_nb}
In the distinguishible case, let $t_1$ and $t_2$ denote two real numbers such that $t_1 \geq 1$ and $t_2 \geq 1$. Then $\Lattice{t_1} = \Lattice{t_2}$ if and only if $t_1 = t_2$.
\end{proposition}
\begin{proof}
Obviously, if $t_1 = t_2$ then $\Lattice{t_1} = \Lattice{t_2}$. Let us prove the converse, i.e.\  $\Lattice{t_1} = \Lattice{t_2}$ implies $t_1 = t_2$ or equivalently if $t_1 \neq t_2$ then $\Lattice{t_1} \neq \Lattice{t_2}$. Suppose that $t_1 < t_2$. There exists $n \in \bbN$ such that $\lfloor t_1 n \rfloor < \lfloor t_2 n \rfloor$. Either $\lfloor t_2 n \rfloor \notin \Lattice{t_1}$, in which case $\Lattice{t_1} \neq \Lattice{t_2}$, or $\lfloor t_2 n \rfloor = \lfloor t_1 n_1 \rfloor$ with $n_1 > n$. In the latter case, distinguishibility implies that there are strictly more elements in $\Lattice{t_1} \cap A$ than in $\Lattice{t_2} \cap A$ where $A$ denotes the set of integers $\interval{ 0, \lfloor t_2 n \rfloor }$.
\end{proof}

% \textbf{Remark}
% Proposition \ref{injective} states that the mapping $k \mapsto x = \lfloor \tau \, k \rfloor$ is injective in the distinguishible case. In the previous result, we prove that the mapping $t \mapsto \Lattice{t}$ is injective.
% 
\subsection{The Set of Compatible Values}

\paragraph{}
Proposition \ref{equiv_grid_nb} is not sufficient for estimating $\tau$ because in practice we only observe a finite set $\SeenInts \varsubsetneq \Lattice{\tau}$. Consequently we define the notion of compatible lattices and equivalently compatible values. For any positive real $t$, we say that $t$ is \emph{compatible} with the data if $\SeenInts \subset \Lattice{t}$. Likewise, for any two sets $A$ and $B$, $A$ is compatible with $B$ if $B \subset A$. Being compatible with the data set is a necessary condition for a valid estimator of $\tau$. 
% Moreover, the set of compatible values hereby denoted $\Compatible{\SeenInts}$ is 

% \paragraph{}
% The true value $\tau$ in the model $X = \lfloor \tau \, N \rfloor$ is always compatible with a sample from $X$. Consequently, we may restrict the search for $\tau$ to the set of compatible values hereby denoted $\Compatible{\SeenInts}$ where $\SeenInts$ are the observed integers. 
% We first describe $\Compatible{A}$ when $A$ is an infinite lattice, then extend the results to a finite lattice.

\paragraph{}
The set of values that are compatible with the infinite lattice $\Lattice{\tau}$ is adequate for estimating $\tau$ because of the following proposition.

\begin{proposition}
\label{proposition1}
$\tau$ is the largest real number in $\Compatible{\Lattice{\tau}}$.
% The set $\Compatible{\Lattice{\tau}}$ of values that are compatible with the (infinite) lattice $\Lattice{\tau}$ is bounded by $\tau$ from above.
\end{proposition}
\begin{proof}
$\tau$ is a compatible value, we only have to show that it is the largest.\\
Let $u$ denote a real number greater than $\tau$, and let $\alpha$ denote a positive real number such that $\tau < \tau + \alpha < u$. We will prove that $u$ is not compatible with $\tau$ by constructing an element in $\Lattice{\tau}$ that cannot be in $\Lattice{u}$. \\
Let $a$ denote a positive integer such that $a> \frac{1}{\alpha}$, and $n = \lfloor \tau a \rfloor$. \\
Suppose that $\Lattice{\tau} \subset \Lattice{u}$ then $n$ belongs to $\Lattice{u}$, and there exists a positive integer $a'$ such that $n = \lfloor u a' \rfloor$. \\
$a' \geq a$ because in the distinguishible case, $a$ and $a'$ correspond to their indices in the sets $\Lattice{\tau}$ and $\Lattice{u}$ and $\Lattice{\tau} \subset \Lattice{u}$. \\
Moreover, as $\tau \leq u$ we have $\lfloor \tau a \rfloor \leq \lfloor u a \rfloor \leq \lfloor u a' \rfloor$. For all three terms to be equal to $n$ in the distinguishible case requires that $a = a'$. \\
Consequently, both $\tau$ and $u$ lie in the interval $[\frac{n}{a},\frac{n+1}{a}[$. As a result, $|\tau-u| \leq \frac{1}{a}$ which contradicts $a > \frac{1}{\alpha}$.
\end{proof}

\paragraph{}
The set $\Compatible{\Lattice{\tau}}$ has an intricate structure. It contains the positive real numbers smaller than 1 and the harmonics $ \left\{ \frac{\tau}{k}, k \in \bbN^* \right\}$, but these are not the only values. For example, $4/3$ is compatible with 2 because every even integer can be written as $\floor{k \times 4/3}$, $k \in \bbN$. Indeed, let $k$ be an even integer. Either $k$ is a multiple of $4$, in which case $k = 4i = \floor{\frac{4}{3} \times 3 i }$, or $k = 4i + 2 = \floor{ \frac{4}{3} \times (3i+2)}$.

\subsection{Estimation with a Finite Lattice}
\label{implementation}

\paragraph{}
In practice, the empirical lattice is finite and can contain missing values and outliers. We say that an integer is \emph{missing} from $\SeenInts$ when it is in the theoretical lattice $\Lattice{\tau}$, smaller than $\maxS = \max \SeenInts$, but not in $\SeenInts$. The set of compatible values with $\SeenInts$ is a finite union of intervals and compatible values are never isolated. As the data contains less information, the true parameter $\tau$ is not the supremum of $\Compatible{\SeenInts}$, but it is still maximal in the following sense.

\begin{proposition}
%Let $\SeenInts$ denote the empirical lattice generated by a sample from the model $X = \floor{\tau N}$. 
\label{prop_algo}
The set of compatible values $\Compatible{\SeenInts}$ contains exactly $]0,1]$ and intervals of length at least $1 / \maxS ^2$ where $\maxS = \max \SeenInts$. In particular, if there are no outliers or missing values in $\SeenInts$ then $\tau$ belongs to the interval $[a,b[$ such that $b = \sup \Compatible{\SeenInts}$.
\end{proposition}

\paragraph{}
The proof is based on the following two lemmas.

\begin{lemma}
\label{proposition_length}
The set of compatible values contains exactly $]0,1]$ and a finite number of intervals of the form $[a,b[$ of length at least $1 / \maxS ^2$ where $\maxS = \max \SeenInts$.
\end{lemma}
\begin{proof}
Let $t>1$ denote a compatible value. For each observed value $x \in \SeenInts$, there exists an integer $n$ such than $x = \floor{t n}$. Consequently, $t$ verifies $t \in [\frac{x}{n},\frac{x+1}{n}[$. The intersection of the constraints $t \in [\frac{x}{n},\frac{x+1}{n}[ $ for all $x \in \SeenInts$ is an interval $t \in [\frac{x_1}{n_1}, \frac{x_2}{n_2}[$. All values $t \in [\frac{x_1}{n_1}, \frac{x_2}{n_2}[$ verify all of the constraints and are thus compatible. The length of this interval is $\frac{x_2}{n_2} - \frac{x_1}{n_1}$ which is at least $1 / (n_1 n_2)$. In the distinguishible case, $n_1 < \max \SeenInts$ and $n_2 < \max \SeenInts$, which implies that the length is at least $1/(\max \SeenInts)^2$.
\end{proof}

\begin{lemma}
Let $t$ be a positive real number that is compatible with the empirical lattice. Then $\displaystyle t < \frac{\maxS + 1}{n}$.
\end{lemma}
\begin{proof}
This follows directly from the upper bounds in Section \ref{upperBound}. See Proposition \ref{proposition_UpperBound} in the Appendix.
\end{proof}

% The set of compatible values is bounded from above by $\tau$ according to Proposition \ref{proposition1} in the theoretical case of an infinite observation lattice. In practice, there may exist compatible values that are larger than $\tau$ but the upper bound given in Proposition \ref{proposition_UpperBound} remains valid for compatible values.

\paragraph{}
To complete the proof, it suffices to show that $\tau$ belongs to the largest interval.

\begin{proof}
There are only finitely many intervals of length at least $1/\maxS^2$ in $[0,\frac{\maxS + 1}{n_{max}}]$, so there exists such an interval $[a,b[$. \\
Let $\mathcal{N}$ denote the set $\mathcal{N} = \{ n | \lfloor n \tau \rfloor \in \SeenInts \}$, i.e.\ the set of values for $N$ that generate $\SeenInts$.
$\tau$ belongs to a certain interval $[a',b'[$ which is the intersection of the constraints $\tau \in [\frac{x}{n},\frac{x+1}{n}[$, for all $x = \floor{ n \tau }$ in $\SeenInts$. We show that $b' = b$, i.e.\ no positive real is both greater than $b'$ and compatible. Let $t$ such that $b' < t$. For all $n \in \mathcal{N}$, $\floor{ n \tau } \leq \floor{ n t }$. As $t \notin [a',b'[$, $t$ breaks at least one of the constraints, that is to say, there is an integer $x$ in $\SeenInts$ such that $x = \lfloor n \tau \rfloor < \lfloor n t \rfloor$. $x$ is skipped in $\Lattice{t}$ and thus $t$ is not compatible.
\end{proof}

\paragraph{}
The previous proposition suggests that it suffices to find the largest compatible interval to estimate $\tau$, and this is our proposed estimator $\tilde{\tau}$. More precisely, the set $\Compatible{\SeenInts} = \cup_{j=1}^J [a_j,b_j[$ is a union of $J$ intervals, with $(a_j)$ and $(b_j)$ increasing sequences, then 
\[
 \tilde{\tau} = \frac{a_J+b_J}{2}.
\]

\paragraph{}
We use the following algorithm to compute $\tilde{\tau}$. This also computes the mapping $x = \lfloor \tau n \rfloor \mapsto n$ and the precision.
\begin{itemize}
 \item compute the set of observed values by sorting the data set and removing multiple occurences
 \item compute the upper bound $\displaystyle \tau < B = \frac{\maxS + 1}{n}$ where $n = \card \SeenInts$
 \item find an approximation of the largest compatible value $t$ by testing the compatibility of the real numbers $\displaystyle t_k = B - \frac{k}{\maxS^2}$
 \item deduce the indexes from $t$, that is to say for all $x \in \SeenInts$, find $i$ such that $x = \lfloor t i \rfloor$
 \item compute the interval $[a,b[$ as the intersection of the constraints $\displaystyle t \in \left[\frac{x}{i},\frac{x+1}{i}\right[$, for all $x$ in $\SeenInts$
 \item return $\displaystyle \frac{a+b}{2}$ as an estimator for $\tau$
\end{itemize}

% However, as $\tau$ is always a compatible value, the estimator cannot return a value that is lower than $a$. Choosing the highest compatible value follows a parsimony principle in which $\tau$ corresponds to the smallest lattice that can explain the dataset.

% Outliers can be dealt with by looking for values that are compatible with all but a few points in the empirical lattice. 

\subsection{Properties of the Estimator}

\paragraph{}
According to the previous results, the estimator performs well when there are no missing values or outliers. Its precision is $(b-a)/2$ and can be computed inside the algorithm. The precision is at least $1/\maxSi$, but depending on the value of $\tau$ it can reach a precision on the order of $1/\maxSi^2$ . In all cases, the precision is better than the density bound, and there is a lower bound.

\paragraph{}
If there are missing values or outliers, the algorithm may find an interval of compatible values that does not contain $\tau$. For example, if the dataset is $\{0,2,4,6,8\}$, a reasonable estimator would answer 2 and not $\tau = 4/3$ with missing values 1 and 5. In practice, such cases are rare, and are related to arithmetic properties of the set $\SeenInts$. However, the largest compatible value is never an erroneous answer to the problem. It is a parcimonious answer in the sense that it is the smallest lattice which may explain the dataset.

% \paragraph{}
% The estimator is optimal in the sense that it uses the information in the dataset completely. This is because all the values in the interval $[a,b[$ are compatible with the dataset and moreover all of them verify the constraints $x = \lfloor t i \rfloor$ so much so that they cannot be distinguished based on $\SeenInts$. 
% 
\paragraph{}
The estimator is optimal in the sense that the algorithm finds an interval of positive real numbers that are all plausible. Given a dataset $(x_1 = \floor{\tau \: i_1}, \ldots, x_n = \floor{\tau \: i_n})$ of size $n$, there is an interval of compatible values that can generate $(x_1,\ldots, x_n)$ from the same realization $(i_1,\ldots,i_n)$ of $N$. Let $[a_J,b_J[$ with $b_J = \sup \Compatible{\SeenInts}$, the following proposition holds.

\begin{proposition}
Given a realization $(i_1,\ldots,i_n)$ of $N$, all values in $[a_J,b_J[$ generate the same data set $(x_1,\ldots,x_n)$, i.e.\
\[
	\forall t \in [a_J,b_J[, \; \forall j \in \interval{1,n}, \; x_j = \floor{\tau \; i_j} = \floor{t \; i_j}
\]
\end{proposition}
\begin{proof}
As in the proof of Proposition \ref{prop_algo}, $[a_J,b_J[$ is the intersection of the constraints $x_j = \floor{t \; i_j}$.
\end{proof}

\paragraph{}
The data set does not contain enough information to distinguish the values in $[a_J,b_J[$. In particular, even if the realization $(i_1, \ldots, i_n)$ is given, then the values are not distinguishible. Note that if $x_0$ is known not to be in $\Lattice{\tau}$, then for all integers $i$, $\tau \geq \frac{x_0+1}{i}$ or $\tau < \frac{x_0}{i}$. These inequalities are not informative because they are already contained in $x = \floor{t i}, \forall x \in \SeenInts$.

\paragraph{}
The program is quite fast. First because is relies only on the set $\SeenInts$ which is much smaller than the dataset when $\tau$ is near 1 and $N$ is independent identically distributed, because repeats of $N$ are discarded. As the following proposition shows, with few missing values, the density bound is precise and the algorithm is quicker. All compatible values can be retrieved by testing $B \maxS ^ 2$ numbers.
\begin{proposition}
If there are no missing values, the largest compatible value is found after at most $\displaystyle \frac{\maxS^2}{\maxSi} \simeq \tau \maxS$ steps. With a small number of missing values $k \ll \maxSi$, the number of steps is on the order of $\displaystyle \tau^2 \maxS \left(k + \frac{1}{\tau}\right)$ where $k = \maxSi - \card \SeenInts$ is the number of missing values.
\end{proposition}
\begin{proof}
% The procedure begins at $B = \frac{\maxS + 1}{n}$ and returns as soon as $t_k < b = (\maxS+1)/\maxSi$. The length of this interval is $B - b = (\maxS+1)(\maxSi - n)/(n \maxSi)$. As the interval is sampled in steps of $1/n^2$, the required number of steps is $\{(\maxS+1)/\maxSi\} \, (\maxSi - n) n$ where $\maxSi - n$ is the number of missing values.
The procedure begins at $B = \frac{\maxS + 1}{n}$, ends before $\frac{\maxS}{\maxSi}$ because $\tau \geq a \geq \frac{\maxS}{\maxSi}$, and proceeds in steps of length $1/\maxS^2$. Consequently, there are at most $C = \maxS^2 ( \frac{\maxS + 1}{n}- \frac{\maxS}{\maxSi})$ steps. Let $k = \maxSi - \card \SeenInts$ denote the number of missing values. We make the following three approximations: $k \ll \maxSi$, $1 \ll \maxS$ and $\tau \simeq \frac{\maxS}{\maxSi}$. Then $C = \maxS^2 (\maxS + 1) \left(\frac{k}{\maxSi \card \SeenInts} + \frac{1}{\maxSi(\maxS+1)}\right)$ which can be approximated by $ C \simeq \tau^2 \maxS (k + \frac{1}{\tau})$.
\end{proof}

\paragraph{}
Testing for the compatibility of a real $t$ is linear in the size of $\SeenInts$, so the whole procedure is at most quadratic. The full set of compatible values can be obtained in cubic time.

% \paragraph{Precision}
% Discuss about extreme value distributions ? \\
% The precision is equal to the length of $I$. This length is lower than the bounds obtained by using the maximum value, so it is an improvement on the density bound. \\
% The precision cannot be easily computed a priori, but the length a posteriori is computed. \\
% \begin{proposition}
% The precision is entirely determined by the rest modulo 1 of $\tau$.
% \end{proposition}
% \begin{proof}
% The precision is determined by the intersection of intervals. We will show that the intervals around $\tau$ are the same as those around $\tau+1$. \\
% Let $\tau \in \bbR$. The values that are in $I$ are those that simultaneously verify the constraints $i = \lfloor k \tau \rfloor$ for all $i$ in the dataset. \\
% Let $i = \lfloor k \tau \rfloor$ denote one such constraint. This is equivalent to $ i / k \leq \tau < (i+1) / k$. \\
% There is a corresponding constraint for $\tau + 1$ which is $i' = \lfloor k (\tau+1) \rfloor$ or equivalently $ i' / k \leq \tau + 1 < (i'+1) / k$. \\
% At this point, we note that $i' = \lfloor k (\tau+1) \rfloor = \lfloor k \tau + k \rfloor = k + \lfloor k \tau \rfloor = k + i$ \\
% In conclusion, the corresponding constraint is $ i / k + 1 \leq \tau + 1 < (i+1) / k + 1$. \\
% This shows that constraints on $\tau$ and $\tau + 1$ are the same.
% \end{proof}
% 
% The precision attains both bounds $1/n_m^2$ and $1/n_m$. (no pf yet)

%%%%%%%%%%%%%%%%%%%%%%%%%%%%%%%%%%%%%%%%%%%%%%%%%%%%%%%%%%%%%%%%%%%%%%%%%%%%%%%%%%%%%%%%%%%%%%%%%%%%
\section{Results and Discussion}
\label{simulations}

% \takehome{
% \begin{itemize}
%  \item the estimator is good because it is precise (absolute precision independent of $\tau$, relative precision is $C/\tau$ which links to the case where the discretization error is negligible)
%  \item description of Fourier estimator : increase sampling, then detect major frequency, maybe harmonics ? not suited to $\tau \sim 1$ because it outputs always 1. May be affected by missing values.
%  \item description of Linear Reg estimator : affected by missing values but can deal with $\tau \sim 1$
%  \item it is not practical to try to rediscover the indexes in the case of missing values
%  \item figure : show the case $\tau \sim 1$
%  \item figure : show the case of missing values (artificial and natural with a Cauchy distrib)
% \end{itemize}
% }

\subsection{Compatible Values Estimator}

% \begin{proposition}
% The precision is entirely determined by the rest modulo 1 of $\tau$.
% \end{proposition}
% \begin{proof}
% The precision is determined by the intersection of intervals. We will show that the intervals around $\tau$ are the same as those around $\tau+1$. \\
% Let $\tau \in \bbR$. The values that are in $I$ are those that simultaneously verify the constraints $i = \lfloor k \tau \rfloor$ for all $i$ in the dataset. \\
% Let $i = \lfloor k \tau \rfloor$ denote one such constraint. This is equivalent to $ i / k \leq \tau < (i+1) / k$. \\
% There is a corresponding constraint for $\tau + 1$ which is $i' = \lfloor k (\tau+1) \rfloor$ or equivalently $ i' / k \leq \tau + 1 < (i'+1) / k$. \\
% At this point, we note that $i' = \lfloor k (\tau+1) \rfloor = \lfloor k \tau + k \rfloor = k + \lfloor k \tau \rfloor = k + i$ \\
% In conclusion, the corresponding constraint is $ i / k + 1 \leq \tau + 1 < (i+1) / k + 1$. \\
% This shows that constraints on $\tau$ and $\tau + 1$ are the same.
% \end{proof}

\paragraph{}
Figure \ref{fig:cve1} illustrates the compatible values estimator on a simulated dataset. The dataset $\{6,6,11,5,3,5,2,6,5,13,2,7,7,7,6\}$ is obtained from the observation model $X = \lfloor 1.32 * N \rfloor$ where $N$ is distributed according to a Poisson random variable with mean $5.5$. It is first reduced to the lattice $\SeenInts = \{2,3,5,6,7,11,13\}$ and is shown at the bottom.

\paragraph{}
The vertical axis represents values of $\tau$. The set of compatible values is composed of several intervals and represented on the left. For each interval, we select one compatible value $t$ and represent the lattice $\Lattice{t}$. All reals in the same interval generate the same lattice, up to $\max(\SeenInts)$.

\begin{figure}
 \centering
 \includegraphics[width=0.45\textwidth]{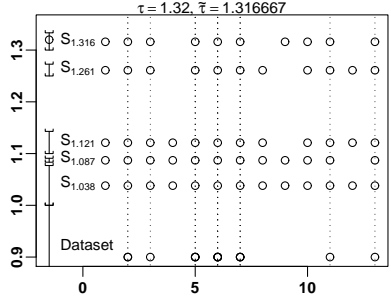}
 % exemple.eps: 0x0 pixel, 0dpi, 0.00x0.00 cm, bb=
 \caption{Comparison of a few compatible lattices and the dataset.}
 \label{fig:cve1}
\end{figure}

\paragraph{}
For comparison, Figure \ref{fig:cve2} displays $\Lattice{t}$ for several values that are not compatible with the data. For example, $5$ and $11$ are in the dataset but not in $\Lattice{1.2}$.

\begin{figure}
 \centering
 \includegraphics[width=0.45\textwidth]{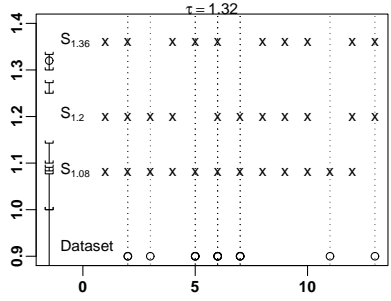}
 % exemple.eps: 0x0 pixel, 0dpi, 0.00x0.00 cm, bb=
 \caption{Comparison of the dataset and a few lattices that are not compatible.}
 \label{fig:cve2}
\end{figure}

\paragraph{}
Two sources of variation affect the estimate $\tilde{\tau}$. First, the estimator is not perfect because the dataset is finite. Second, the data set $\SeenInts$ is random. Figure \ref{fig:cve3} shows the performance of the estimator with a fixed dataset ($N \in \interval{1,10}$) for several values of $\tau$. The intervals shown correspond to the intervals in $\Compatible{\SeenInts}$ that contain the largest compatible value.
\begin{figure}
 \centering
 \includegraphics[width=0.45\textwidth]{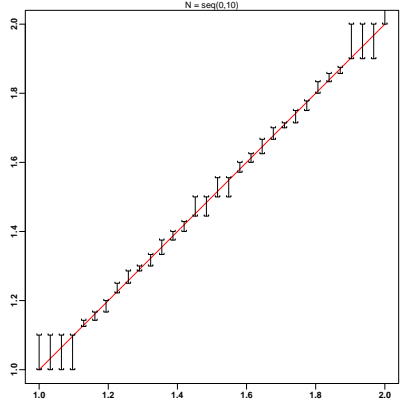}
 % exemple.eps: 0x0 pixel, 0dpi, 0.00x0.00 cm, bb=
 \caption{Length of the maximal interval for several values of $\tau$ with $N \in \interval{1,10}$.}
 \label{fig:cve3}
\end{figure}

\paragraph{}
We can see that the precision of the estimator varies with $\tau$. Only the range $[1,2]$ is shown because the precision only depends on the rest $\tau - \lfloor \tau \rfloor$ modulo 1. Consequently, the absolute precision is roughly constant, whereas the relative precision is $\displaystyle O\left(\frac{1}{\tau}\right)$. With small quantization error ($\tau \gg 1$) the estimation problem is easier.

% \textbf{Useful ?} Longueur moyenne de l'intervalle ? (c'est réalisable car la longueur est bornée donc intégrable)

\paragraph{}
Figure \ref{fig:cve4} shows the distribution of $\tilde{\tau}$ when the dataset $\SeenInts$ is the result of 15 samples of $X = \floor{\tau N}$ where $\tau = 1.32$ and $N$ is distributed according to a Poisson random variable with mean $5.5$. The distribution of the estimator value is obtained from the 200 repeats shown in the bottom of the plot thanks to a kernel estimate, even if the distribution is a sum of Dirac point masses. The interval shows the interval obtained with the (complete) dataset $\{ \floor{1.32*n}, n \in \interval{1,13} \}$.
\begin{figure}
 \centering
 \includegraphics[width=0.45\textwidth]{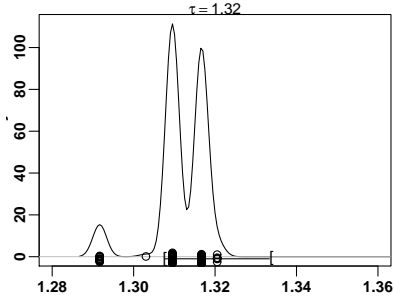}
 % exemple.eps: 0x0 pixel, 0dpi, 0.00x0.00 cm, bb=
 \caption{Kernel density estimate of the distribution of the compatible values estimator on a random dataset.}
 \label{fig:cve4}
\end{figure}

\subsection{The Double Poisson Family}
\label{doublepoisson}

\paragraph{}
In this section, we briefly recall the results from \cite{efron1986def}, and deduce an estimator for our model.
Let $g_\mu(y) = e^{-\mu} \mu^y / y!$ denote the distribution function of a Poisson random variable with mean $\mu$. The double Poisson distribution with parameters $\theta,\mu$ is defined as:
\begin{align*}
  f_{\theta,\mu} (y) & = c(\theta,\mu) \theta ^ {1/2} \left\{ g_\mu(y) \right\}^\theta \left\{ g_y(y) \right\}^{1-\theta} \\
                     & = c(\theta,\mu) \left(\theta ^ {1/2} e^{-\theta \mu}\right) \left( \frac{e^{-y} y^y}{y!} \right) \left( \frac{e \mu}{y} \right)^{\theta y}
\end{align*}
where $c$ is a normalization constant.

\paragraph{}
Maximum Likelihood Estimation leads to the following estimators. Let $(Y_1, \ldots, Y_n)$ be independent identically distributed random variables with distribution $f_{\theta,\mu}$, then 
\begin{align*}
 \hat{\mu}    & = \frac{1}{n} \sum_{i=1}^n Y_i \\
 \hat{\theta} & = \frac{n}{2 \sum_{i=1}^n I(Y_j,\mu)}
\end{align*}
where $I(\mu_1,\mu_2) = \mu_1 (\log(\mu_1) - \log(\mu_2)) - (\mu_1 - \mu_2)$.

\paragraph{}
Let $Y_{\theta,\mu}$ be a random variable with distribution function $f_{\theta,\mu}$, then according to \cite{efron1986def} $Y_{\theta,\mu}$ has approximately the same distribution as $X / \theta$ where $X$ is Poisson distributed with mean $\mu \theta$. With Poisson distribution for the ion counts, our observation model becomes $Y = \floor{\tau N}$ where $N$ is Poisson distributed with mean $\lambda$. Consequently, estimates for $\tau$ and $\lambda$ can be deduced from $\hat{\theta}$ and $\hat{\mu}$ with the following relations:
\begin{align*}
 \hat{\tau}    & = \frac{1}{\hat{\theta}} \\
 \hat{\lambda} & = \hat{\mu} \hat{\theta}
\end{align*}

\paragraph{}
The double Poisson distribution is a correct approximation of the distribution of $X / \theta$ for large $\mu$, and in that case, $\hat{\tau}$ and $\hat{\lambda}$ are unbiased estimates of $\tau$ and $\lambda$. The standard deviation of $\hat{\tau}$ is $\frac{\tau \sqrt{2}}{\sqrt{n}}$. Figure \ref{fig:dpois1} shows the distribution of $\hat{\tau}$ with flooring noise, i.e.\ in the model $Y = \floor{\tau N}$ (solid line) and without flooring noise in the model $Y = \tau N$ (dotted line). The plot was generated with 2000 repeats with data sets of size 500. We observe a large standard deviation compared with the compatible values estimator, even on a much larger data set. With large values of $\mu$, truncation has limited effect on the estimate.

\begin{figure}
 \centering
 \includegraphics[width=0.45\textwidth]{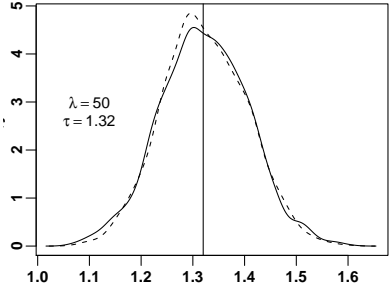}
 % exemple.eps: 0x0 pixel, 0dpi, 0.00x0.00 cm, bb=
 \caption{Kernel density estimate of the distribution of the double poisson estimate of $\tau$ with flooring noise (solid line) and without (dotted line).}
 \label{fig:dpois1}
\end{figure}

\paragraph{}
For modeling rare ion count events, we need to study the estimators with small values of $\lambda$ and $\tau$. In that case, $\hat{\tau}$ is strongly biased for both models as shown on Figure \ref{fig:dpois2}. This implies that the approximation is not suited to this range of parameters, and that the flooring noise makes a significant difference there. Figure \ref{fig:dpois2} was generated using 2000 repeats with data sets of size 500. For comparison, we show the optimal interval obtained by the compatible values estimator on the data set $\{ \floor{1.32*n}, n \in \interval{1,13} \}$.

\begin{figure}
 \centering
 \includegraphics[width=0.45\textwidth]{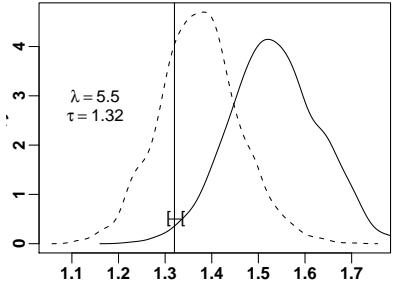}
 % exemple.eps: 0x0 pixel, 0dpi, 0.00x0.00 cm, bb=
 \caption{Kernel density estimate of the distribution of the double poisson estimate of $\tau$ with flooring noise (solid line) and without (dotted line).}
 \label{fig:dpois2}
\end{figure}

% grandi0 = function(m1,m2) {
%    if (m1 == 0) 
%       return(m2)
%    else
%       return(m1*(log(m1)-log(m2)) - (m1 - m2))
% }
% grandi = Vectorize(grandi0)
% 
% tau = 1.32
% lambda = 15
% data = floor(tau * rpois(1000,lambda))
% 
% mu_hat = mean(data)
% theta_hat = length(data) / 2 / sum(grandi(data,mu_hat))
% lambda
% mu_hat*theta_hat
% tau
% 1/(theta_hat)

\subsection{Fourier Estimator}

\paragraph{}
From the set $\tau \bbN$ we can construct the signal $f : t \mapsto \sum_{k \in \bbN} \delta(x-\tau k)$ where $\delta$ denotes the Dirac function, that is to say a periodic series of pulses. The period $\tau$ may thus be estimated using Fourier transform. Likewise, we define the estimator $1 / \tau_F$ as the maximum of the Fourier transform of the quasi-periodic signal $f : t \mapsto \sum_{k \in \bbN} \delta(x-\floor{\tau k})$. 

\paragraph{}
As $\tau$ can be seen as a quasi-period, our estimation problem is closely linked to the ``harmonic retrieval problem''. Many approaches have been proposed in that domain and the main focus is on the estimation of the Power Spectral Density \cite{hayes1996sds}. However, the signal is usually perturbed by additive noise whereas in this paper we consider a distortion of the time axis.
% Référence aux radars ?

\paragraph{}
We use the following algorithm:
\begin{itemize}
 \item sample the signal $f$ at the points $x_i = i$ for the integers $i$ in $\interval{0,\max(\SeenInts)}$
 \item compute the Discrete Fourier Transform
 \item compute an upper bound using Proposition \ref{proposition_UpperBound} : $\displaystyle \tau < B = \frac{\maxS + 1}{n}$
 \item find the frequency with highest absolute Fourier coefficient
 \item return the corresponding period (inverse of the frequency)
\end{itemize}

\paragraph{}
This estimator has a precision that corresponds to the sampling rate in time space around the true value. In the Fourier space, the sampling rate is uniform with steps of length $1 / \max(\SeenInts)$ which is equivalent to $1 / (\tau \times \maxS)$. In the time space, as $P = 1/f$, then $\Delta P = - \Delta f / f^2$ and the sampling rate is non uniform. For $f = 1/\tau$ we obtain the precision of the Fourier estimator as $\tau / \maxS$. This suggests that the precision decreases with $\tau$. However, the signal frequency $1/\tau$ is near $\maxS / \max(\SeenInts)$ which is one of the sampling points. As a result, in practice, the absolute precision is on the order of $1/\maxS$ and independent of $\tau$.

\paragraph{}
Figure \ref{fig:fourier1} shows in the frequency and period space the Fourier transform of the quasi-periodic signal obtained from the dataset $N \in \interval{1,10}$. The vertical line corresponds to the upper bound from Proposition \ref{proposition_UpperBound}.

\begin{figure}
 \centering
 \includegraphics[width=0.45\textwidth]{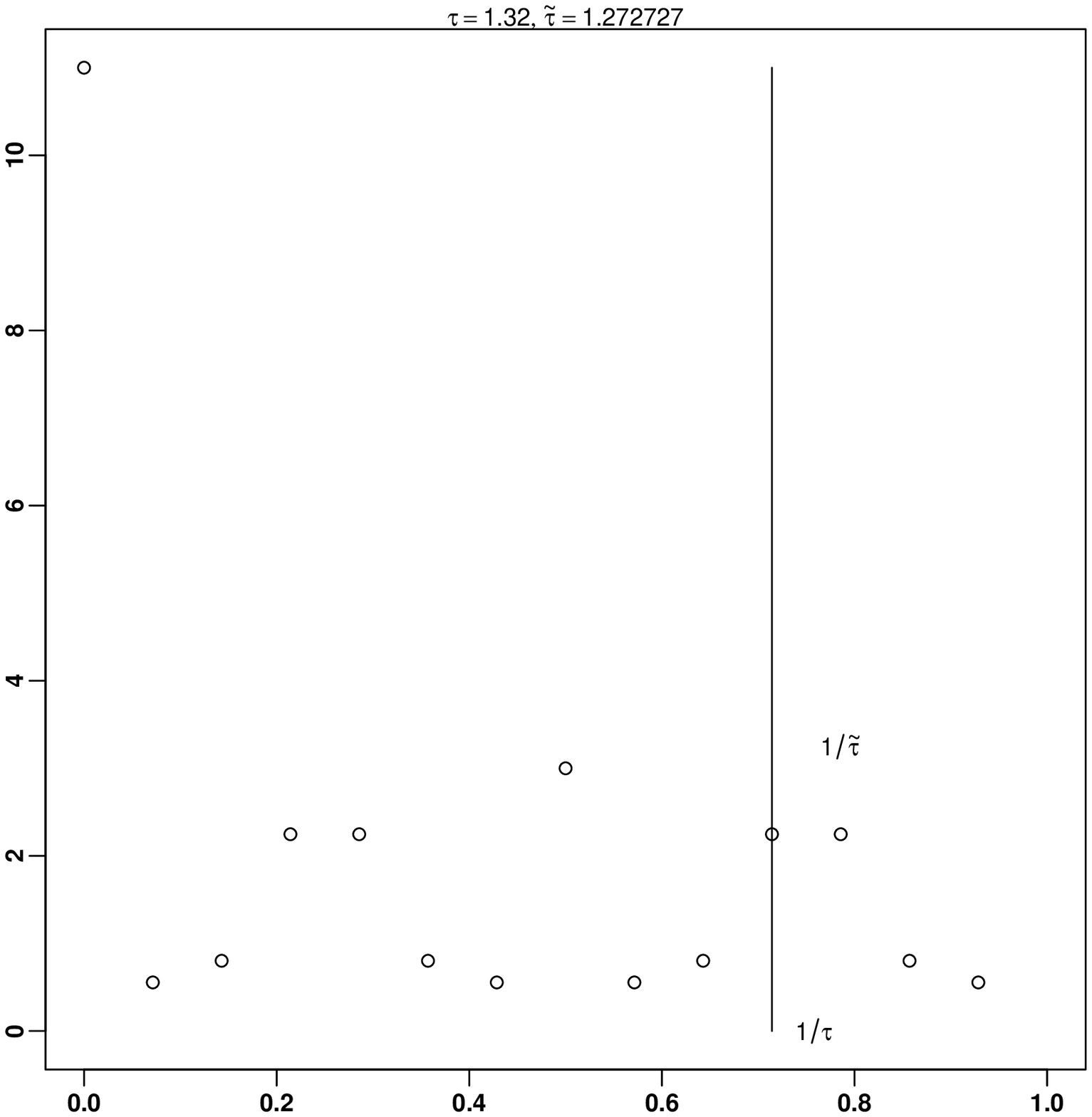}
 \includegraphics[width=0.45\textwidth]{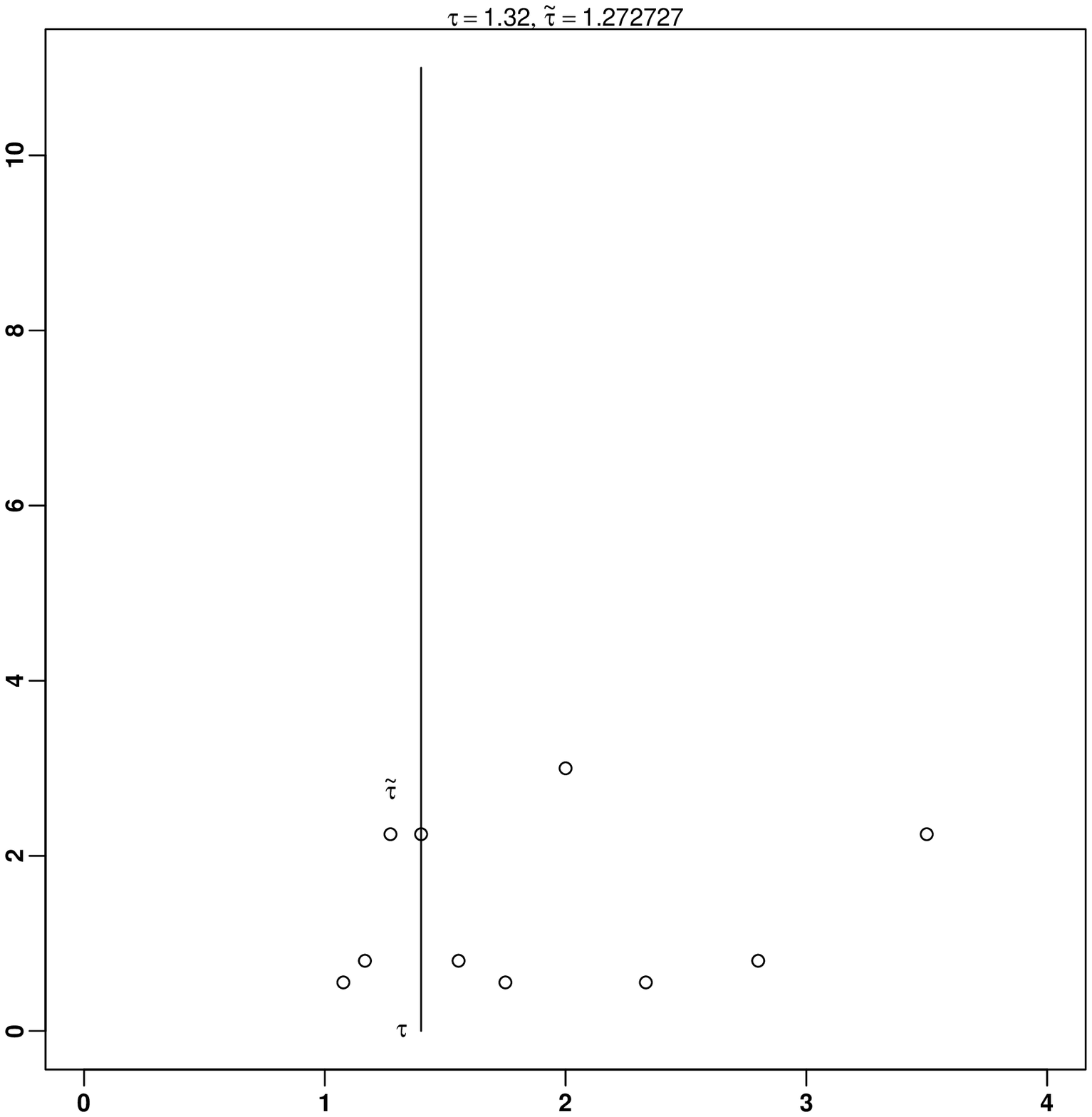}
 % exemple.eps: 0x0 pixel, 0dpi, 0.00x0.00 cm, bb=
 \caption{Fourier transform of the quasi-periodic signal, in Fourier space (left) and period space (right). The vertical line shows the upper bound from Proposition \ref{proposition_UpperBound}.}
 \label{fig:fourier1}
\end{figure}

\textbf{Remark}
When oversampling by a factor $k$, i.e.\ sampling at the points $\displaystyle x_i = \frac{i}{k}$ for the integers $i$ in $\interval{0,\max(\SeenInts) \times k}$, the harmonics of 1 Hz increase in magnitude. Therefore it is necessary to weed out the frequencies above 1 Hz in the distinguishible case. Moreover, oversampling increases the maximum frequency that can be represented in the Fourier space and does not improve the precision of the estimator.
% This has been verified. Changes in scale may affect bias

\paragraph{}
On a random dataset, the Fourier estimator suffers greatly from missing values. Figure \ref{fig:fourier2} shows the distribution of $\tau_F$ with 200 simulations and a dataset of size 15 where $N$ is distributed according to a Poisson random variable with mean $5.5$. The precision of the estimator is much worse than the compatible values estimator (see the plotted interval). The Fourier estimate $\tau_F$ is compatible with the dataset in only about 1\% of the simulations.
\begin{figure}
 \centering
 \includegraphics[width=0.45\textwidth]{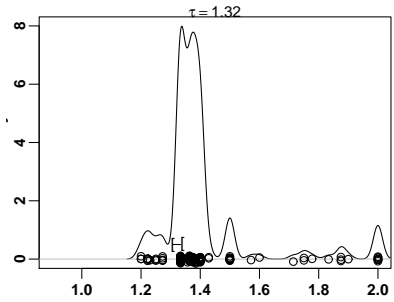} 
 % exemple.eps: 0x0 pixel, 0dpi, 0.00x0.00 cm, bb=
 \caption{Kernel density estimate of the distribution of the Fourier estimator.}
 \label{fig:fourier2}
\end{figure}

%\textbf{Remark}
%A lot of the difficulties in this estimator are related to the true frequency not being the maximum of the Fourier transform. Some of them are alleviated by using an upper bound for the period, which corresponds to a lower bound on the frequency. However, it is difficult to choose between a frequency and its alias

\subsection{Linear Regression Estimator}

\paragraph{}
The observation model $X = \floor{\tau N}$ may be written $X = \tau N + \varepsilon$ where $\varepsilon$ is an error term. Even if $\varepsilon$ is not Gaussian, linear regression can yield a reasonable estimate of the regression coefficient $\tau$ as Figure \ref{fig:linreg1} shows.

\paragraph{}
We use the following algorithm:
\begin{itemize}
 \item compute the empirical lattice $\{ x_i \}$ by sorting and removing duplicates in the dataset
 \item compute the indexes $ \{ n_i \}$ according to the sorting index
 \item fit a regression line of the form $x_i = a n_i + 0.5$
 \item return $a$
\end{itemize}

\paragraph{}
Figure \ref{fig:linreg1} shows the linear regression estimator on the dataset $\SeenInts = \floor{\tau \interval{1,10}} $ (no missing values). For each element in the dataset, if the regression line intersects the length 1 interval then the estimate is compatible with the data point.
\begin{figure}
 \centering
 \includegraphics[width=0.45\textwidth]{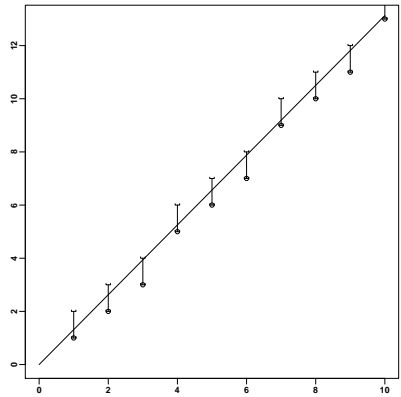}
 % exemple.eps: 0x0 pixel, 0dpi, 0.00x0.00 cm, bb=
 \caption{The linear regression estimator on a dataset without missing values.}
 \label{fig:linreg1}
\end{figure}

\paragraph{}
Note that the truncation error is not centered. Consequently, we compute the regression coefficient in the the model $X + 0.5 = \tau N + \varepsilon$. For the same reason, the regressors are below the regression line.

\paragraph{}
The main difficulty in the linear regression is that the values of the regressor variable $N$ are unknown. In the distinguishible case, it is possible to reconstruct them when there are no missing values, i.e.\ $\Lattice{\tau} \cap \interval{0,n} = \SeenInts$ where $n = \max{\SeenInts}$. Otherwise, the regressors will be shifted and that affects strongly the estimate. Figure \ref{fig:linreg2} shows such a case. The regressors inferred in the linear regression estimator and the regression line are shown in solid line. For comparison, the true regressors are displayed in dotted line. The compatible values estimator finds the true regressors and its regression line is shown in dotted line.

\begin{figure}
 \centering
 \includegraphics[width=0.45\textwidth]{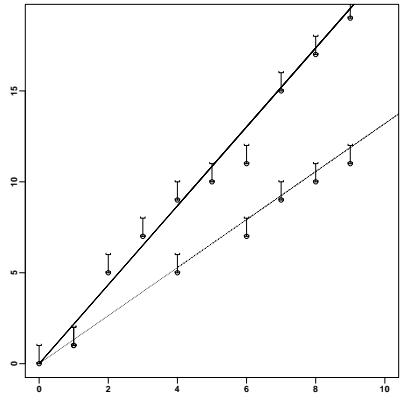}
 % exemple.eps: 0x0 pixel, 0dpi, 0.00x0.00 cm, bb=
 \caption{The linear regression estimator in the case of missing values. The compatible values estimator is shown in dotted line.}
 \label{fig:linreg2}
\end{figure}

\section{Conclusion}

\paragraph{}
In the observation model $X = \floor{\tau N}$, the parameter $\tau$ can be reliably estimated independently from $N$. This allows the full recovery of the statistics of $N$ prior to modeling. The structure of $N$ may then be studied at length afterwards.

\paragraph{}
The estimator based on compatible values is optimal and reasonably quick. It is resistant to missing values in practice, and in the worst case returns an acceptable (parcimonious) answer without hypotheses on the law of $N$. 

\paragraph{}
Unfortunately, this estimator only takes into account truncation noise, and yields poor results on real data. We are currently pursuing an extension of the model that mixes electronic noise and truncation effects.

\paragraph{}
Compared to the other three estimators, the compatible values estimator performs much better but also more slowly. The Double Poisson Family is simply not a suitable model in our range of parameters, but there is room for improvement for the other estimators. For example, the main difficulty in the linear regression is computing the indexes. With some knowledge about the law of $N$, quantile regression could be applied. 

\paragraph{}
The Fourier estimator suggests a strong relationship with the harmonic retrieval problem, although the signal is not periodic. Although the truncation error considered in this paper is very different from Gaussian errors usually considered in harmonic retrieval, some algorithms from that field may make a better compromise between speed and precision for the current problem.

\bibliographystyle{plain} 
\bibliography{EstimationGainParameter}

\section*{Appendix}

\begin{proposition}
\label{injective}
% The model $X = \lfloor \tau \, N \rfloor$ is distinguishible if and only if $\tau \geq 1$.
The mapping $x \mapsto \floor{ \tau \, x }$ is injective if and only if $\tau \geq 1$.
\end{proposition}
\begin{proof}
If $\tau = 1$, the mapping is the identity function. Suppose $\tau > 1$ and let $n_1$ and $n_2$ denote two (positive) integers such that $n_1 < n_2$. Then $\tau n_2 - \tau n_1 > \tau > 1$ and $\floor{\tau n_2} > \floor{\tau n_1}$. When $\tau < 1$, the mapping is not injective because $\floor{\tau \times 1} = \floor{\tau \times 0} =  0$. 
\end{proof}

\subsection*{Upper Bounds on $\tau$}

\begin{proposition}[Any two observations]
\label{any_two_obs}
Let $x$ and $y$ be two distinct elements of the set $\SeenInts$ of observed values. Then $\tau < 1 + |x - y|$.
\end{proposition}
\begin{proof}
Let $i$ and $j$ be the values of $N$ corresponding to $x$ and $y$ i.e.\ $x = \lfloor \tau i \rfloor$ and $y = \lfloor \tau j \rfloor$. Then we have the inequalities:
$x \leq \tau i < x+1$, $y \leq \tau j < y+1$, and thus $\tau (j-i) < y - x + 1$. Assuming $x < y$, we obtain $ \tau < \frac{y-x+1}{j-i} < y-x+1$.
\end{proof}

\begin{proposition}[Observed intervals]
Let $\interval{x,y}$ denote the set of integers between $x$ and $y$. If $\interval{x,y}$ is a subset of $\SeenInts$, then $\displaystyle \tau < 1 + \frac{1}{y-x}$.
\end{proposition}
\begin{proof}
Using the same notations as in the proof of Proposition \ref{any_two_obs}, $\tau < \frac{y-x+1}{j-i} < \frac{y-x}{j-i} + \frac{1}{j-i} < 1 + \frac{1}{j-i}$ because in the distinguishible case the number of elements in $\interval{x,y}$ is $y-x+1 = j-i+1$.
\end{proof}

\begin{proposition}[Density Upper Bound]
\label{proposition_UpperBound}
Let $\maxS = \lfloor \tau \maxSi \rfloor$ denote the largest integer in $\SeenInts$. Then $\displaystyle \tau < \frac{\maxS + 1}{\maxSi}$. When $\maxSi$ is unknown (because of potential missing values), let $n$ denote the number of non zero observed integers. Then $\displaystyle \tau < \frac{\maxS + 1}{\maxSi} \leq \frac{\maxS + 1}{n}$.
\end{proposition}

\end{document}